\documentclass[reqno,a4paper]{amsart}
\usepackage{amsmath}
\usepackage{amsthm}
\usepackage{amssymb}
\usepackage{graphicx}
\usepackage{color}
\usepackage{comment}
\usepackage{bm}
\usepackage[dvipdfm,bookmarks=true,bookmarksnumbered=true,colorlinks=true]{hyperref}
\usepackage{listings}
\usepackage{accents}
\newtheorem{thm}{Theorem}[section]
\newtheorem{prop}[thm]{Proposition}
\newtheorem{lemma}[thm]{Lemma}

\newtheorem{remark}[thm]{Remark}
\numberwithin{equation}{section}
\newcommand{\sn}{\hspace{0.1em}{\rm sn}\hspace{0.1em}}
\newcommand{\cn}{\hspace{0.1em}{\rm cn}\hspace{0.1em}}
\newcommand{\dn}{\hspace{0.1em}{\rm dn}\hspace{0.1em}}

\newcommand{\sz}{{\rm sz}}
\newcommand{\cz}{{\rm cz}}
\newcommand{\dz}{{\rm dz}}
\newcommand{\hsz}{{\rm \widehat{sz}}}
\newcommand{\hcz}{{\rm \widehat{cz}}}
\newcommand{\hdz}{{\rm \widehat{dz}}}
\newcommand{\sge}{{\rm sg}_e}
\newcommand{\cge}{{\rm cg}_e}
\newcommand{\dge}{{\rm dg}_e}
\newcommand{\sgo}{{\rm sg}_o}
\newcommand{\cgo}{{\rm cg}_o}
\newcommand{\dgo}{{\rm dg}_o}
\newcommand{\ii}{{\rm i}}
\newcommand{\eqOmega}{\sim_{\overline{\Omega}}}
\newcommand{\eqH}{\sim_H}
\newcommand{\WH}{W(E_8)/\hspace{-0.3em}\sim_H}
\newcommand{\utilde}[1]{\underaccent{\tilde}{#1}}
\begin{document}
\title[Geometry of an elliptic-difference equation]{Geometry of an elliptic-difference equation related to Q4}
\author{James Atkinson}
\address{Department of Mathematics and Information Sciences, Northumbria University, Newcastle upon Tyne, UK}
\email{mailto:james.atkinson@northumbria.ac.uk}
\author{Phil Howes}
\address{School of Mathematics and Statistics F07, The University of Sydney, NSW 2006, Australia}
\email{howespt@gmail.com}
\author{Nalini Joshi}
\address{School of Mathematics and Statistics F07, The University of Sydney, NSW 2006, Australia}
\email{nalini.joshi@sydney.edu.au}
\author{Nobutaka Nakazono}
\address{School of Mathematics and Statistics F07, The University of Sydney, NSW 2006, Australia}
\email{nobua.n1222@gmail.com}
\begin{abstract}
In this paper, we investigate a nonlinear non-autonomous elliptic difference equation, which was constructed by Ramani, Carstea and Grammaticos by integrable deautonomization of a periodic reduction of the discrete Krichever-Novikov equation, or Q4. 
We show how to construct it as a birational mapping on a rational surface blown up at eight points in $\mathbb P^1\times \mathbb P^1$, and find its affine Weyl symmetry, placing it in the geometric framework of the Painlev\'e equations.
The initial value space is ell-$A_0^{(1)}$ and its symmetry group is $W(F_4^{(1)})$.
We show that the deautonomization is consistent with the lattice-geometry of Q4 by giving an alternative construction, which is a reduction from Q4 in the usual sense.
A more symmetric reduction of the same kind provides another example of a second-order integrable elliptic difference equation.
\end{abstract}

\maketitle
\section{Introduction}
In this paper, we investigate the difference equation
\begin{equation}\label{eqn:rcg1}
\begin{split}
 &{\rm cn}(\gamma_n){\rm dn}(\gamma_n)\big(1-k^2{\rm sn}^4(z_n)\big)u_n\big(u_{n+1}+u_{n-1}\big)\\
 &-{\rm cn}(z_n){\rm dn}(z_n)\big(1-k^2{\rm sn}^2(z_n){\rm sn}^2(\gamma_n)\big)\big(u_{n+1}u_{n-1}+{u_n}^2\big)\\
 &+\big({\rm cn}^2(z_n)-{\rm cn}^2(\gamma_n)\big){\rm cn}(z_n){\rm dn}(z_n)\big(1+k^2{u_n}^2u_{n+1}u_{n-1}\big)=0,
\end{split}
\end{equation}
where $n\in{\mathbb Z}$ is the independent variable, $u_n$ is the dependent variable,
$k$ is the modulus of the elliptic sine, and $\gamma_e$,
$\gamma_o$ are constant complex parameters with
\begin{equation}\label{eqn:constraints}
 z_n=(\gamma_e+\gamma_o)n+z_0,\quad
 \gamma_n=\begin{cases}\gamma_e, \quad\textrm{for}\ n=2\,j\\ \gamma_o,\quad\textrm{for}\ n=2j+1.\end{cases}
\end{equation}
Equation \eqref{eqn:rcg1} was derived by Ramani, Carstea and Grammaticos \cite{RCG2009:MR2525848}. 
They applied singularity confinement to construct this non-autonomous difference equation 
from an autonomous reduction of a partial difference equation first provided by Adler \cite{AdlerVE1998:MR1601866}, 
which is also known as Q4 in the list of such equations that were classified by Adler, Bobenko and Suris (ABS) 
\cite{ABS2003:MR1962121}. (See Section \ref{subsection:staircase_reduction} for details.) 

In this paper, we study the initial value space of Equation \eqref{eqn:rcg1}. 
This is the space of initial conditions in the sense proposed by Okamoto \cite{OkamotoK1979:MR614694} 
for differential Painlev\'e equations and Sakai \cite{SakaiH2001:MR1882403} for discrete Painlev\'e systems. 
Sakai showed that the initial value space of each discrete Painlev\'e equation is a rational surface and 
showed how to construct their symmetry groups as affine Weyl groups 
orthogonal to the divisor class of their initial value space in the Picard lattice. 
However, Equation \eqref{eqn:rcg1} does not explicitly appear in Sakai's list of discrete Painlev\'e equations.
The first aim of the present paper is to identify where it fits in Sakai's classification of Painlev\'e systems.

We show that the initial value space of Equation \eqref{eqn:rcg1} can be identified as the ell-$A_0^{(1)}$ surface, 
i.e., the elliptic surface in Sakai's classification\cite{SakaiH2001:MR1882403}. 
However, Equation \eqref{eqn:rcg1} differs from Sakai's elliptic discrete Painlev\'e equation 
associated with rational surface $A_0^{(1)}$ in an important sense. 
Sakai's equation has affine Weyl symmetry group $W(E_8^{(1)})$, however, 
we find that the affine Weyl symmetry group for Equation \eqref{eqn:rcg1} is  $W(F_4^{(1)})$. 

It is interesting to note that the initial value space of Equation \eqref{eqn:rcg1} is different from the one 
which arises for the autonomous reduction of Q4, i.e., Equation \eqref{eqn:aut_Q4eqn}. 
In the latter case, the initial value space is type $q$-$A_1^{(1)}$, 
which contains eight base points arranged in pairs (with four lying on each of two lines). 
In contrast, for Equation \eqref{eqn:rcg1} the base points lie on an elliptic curve and 
the initial value space is type ell-$A_0^{(1)}$.
This observation is surprising, as it shows that the process of singularity confinement 
changes the initial value space in a non-trivial way.

Also in this paper we investigate the relationship between Equation \eqref{eqn:rcg1} and Q4.
The deautonomisation of the original construction turns out to be consistent with the standard lattice geometry of Q4, provided we take into account the tetrahedral symmetry of the equation.
This allows to reverse-engineer an alternative reduction procedure, establishing how the equation \eqref{eqn:rcg1} defines a particular class of solutions for Q4.
One consequence is that it allows to identify the trivial solutions \cite{MSY2003:MR1958273} of \eqref{eqn:rcg1} with the singular solutions of Q4 that are compatible with the reduction. 
The new reduction procedure admits some natural generalisations, and we give one in particular, because it leads to a second-order integrable equation (Equation \eqref{eqn:todatype2}) in the same class as Equation \eqref{eqn:rcg1}.

The paper is organised as follows. 
In Section \ref{section:geometry}, we construct the initial value space of Equation \eqref{eqn:rcg1}.
Its symmetry group is constructed in Section \ref{section:symmetry}. 
The alternative construction as a direct reduction of Q4, and the trivial solutions, are established in Section \ref{section:solution}. 

\subsection{Periodic reduction of Q4}\label{subsection:staircase_reduction}
In this section, we recall how to obtain Equation \eqref{eqn:rcg1} from Q4 in the Jacobi form \cite{HietarintaJ2005:MR2217106}:
\begin{equation}\label{eqn:Q4eqn}
{\rm sn}\,\alpha\,(u\tilde{u}+\hat{u}\hat{\tilde{u}})-{\rm sn}\,\beta\,(u\hat{u}+\tilde{u}\hat{\tilde{u}})-{\rm sn}(\alpha-\beta)\big(\tilde{u}\hat{u}+u\hat{\tilde{u}}-{\rm sn}\,\alpha\,{\rm sn}\,\beta\,(1+k^2u\tilde{u}\hat{u}\hat{\tilde{u}})\big)=0,
\end{equation}
where $u=u(m,n)$, $\hat{u}=u(m,n+1)$, $\tilde{u}=u(m+1,n)$, $n$ and $m$ are integers, 
$\alpha$ and $\beta$ are constant complex parameters of the equation, 
and $k$ is the modulus of the elliptic sine.
The simple $(1,1)$ periodic reduction $\tilde{u}=\hat{u}$ of Equation \eqref{eqn:Q4eqn} was carried out first by Joshi {\it et al.} \cite{JGTR2006:MR2271126}, who found the following second order equation:	
\begin{equation}\label{eqn:aut_Q4eqn}
\begin{split}
 &({\rm sn}\,\alpha\,-{\rm sn}\,\beta\,)u_n(u_{n+1}+u_{n-1})-{\rm sn}(\alpha-\beta)(u_{n+1}u_{n-1}+{u_n}^2)\\
 &+{\rm sn}\,\alpha\,{\rm sn}\,\beta\,{\rm sn}(\alpha-\beta)(1+k^2{u_n}^2u_{n+1}u_{n-1})=0.
\end{split}
\end{equation}
Ramani {\it et al.} \cite{RCG2009:MR2525848} deautonomised Equation \eqref{eqn:aut_Q4eqn} by the method of singularity confinement after a change of variables $\alpha\rightarrow \gamma+z,\, \beta\rightarrow \gamma -z$.
The resulting equation then becomes Equation \eqref{eqn:rcg1}.

Ramani {\it et al.} \cite{RCG2009:MR2525848} also numerically investigated the growth of degree of the iterates of Equation \eqref{eqn:rcg1} as functions of initial values. 
They found that the growth is quadratic and that the algebraic entropy \cite{BV1999:MR1704282} is zero.  
These results led these authors to assert that Equation \eqref{eqn:rcg1} should therefore be considered to be a discrete Painlev\'e equation. 
We establish how this equation fits into algebro-geometric classification of Sakai by describing initial value space and its symmetry group.

\section{Geometry of System \eqref{eqn:rcgsys}}\label{section:geometry} 
Sakai \cite{SakaiH2001:MR1882403} showed that discrete Painlev\'e systems are bi-holomorphic mappings 
on rational surfaces obtained by an 8-point blow up of $\mathbb{P}^1\times\mathbb{P}^1$. 
In this section, we find these base points of System \eqref{eqn:rcgsys}, 
describe the corresponding initial value space and construct its Cremona isometries.

We start by recasting Equation \eqref{eqn:rcg1} as a system
\begin{subequations}\label{eqn:rcgsys}
\begin{align}
 &\bar{g}
 =\frac{(1-k^2\sz^4)\cge\dge fg-(\cge^2-\cz^2){\rm cz}\,\dz-(1-k^2\sge^2 \sz^2)\cz\,\dz f^2}
 {k^2(\cge^2-\cz^2)\cz\,\dz\, f^2 g-(1-k^2 \sz^4)\cge\dge f+(1-k^2\sge^2\sz^2)\cz\,\dz\,g},\\
 &\bar{f}
 =\frac{(1-k^2\hsz^4)\cgo\dgo \bar{g}f-(\cgo^2-\hcz^2)\hcz\,\hdz-(1-k^2\sgo^2 \hsz^2)\hcz\,\hdz\,\bar{g}^2}
 {k^2(\cgo^2-\hcz^2)\hcz\,\hdz\, \bar{g}^2f-(1-k^2 \hsz^4)\cgo\dgo \bar{g}+(1-k^2\sgo^2 \hsz^2)\hcz\,\hdz\, f},
\end{align}
\end{subequations}
where 
\begin{align}
& \sz=\sn(z_{2n}),\quad
 \cz=\cn(z_{2n}),\quad
 \dz=\dn(z_{2n}),\\
 &\hsz=\sn(z_{2n+1}),\quad
 \hcz=\cn(z_{2n+1}),\quad
 \hdz=\dn(z_{2n+1}),\\
 &\sge=\sn(\gamma_e),\quad
 \cge=\cn(\gamma_e),\quad
 \dge=\dn(\gamma_e),\\
 &\sgo=\sn(\gamma_o),\quad
 \cgo=\cn(\gamma_o),\quad
 \dgo=\dn(\gamma_o),\\
 &f=f_n, \quad 
 g=g_n,\quad
 z_k=(\gamma_e+\gamma_o)k+z_0,
\end{align} 
and $\,\bar{}\,$ means $n\rightarrow n+1$. 
This provides a mapping of $\mathbb{P}^1\times\mathbb{P}^1$ to itself, which we denote by $\varphi_a:(f,g)\rightarrow (\bar{f},\bar{g})$.
Here, the variables $f, g$ are related to $u_n$ by
\begin{equation}
 f_n=u_{2n},\quad g_n=u_{2n-1}.
\end{equation}
\subsection{Initial value space}
System \eqref{eqn:rcgsys} has base points (where the system is ill defined because it approaches $0/0$), which are given by direct calculation:
\begin{subequations}\label{eqn:bps}
\begin{align}
 &p_1: (f,g)=\left(\sn(z_{2n}+K), \sn(\gamma_e+K)\right),\\
 &p_2: (f,g)=\left(-\sn(z_{2n}+K), -\sn(\gamma_e+K)\right),\\
 &p_3: (f,g)=\left(\sn(\gamma_o+K), \sn(z_{2n-1}+K)\right),\\
 &p_4: (f,g)=\left(-\sn(\gamma_o+K), -\sn(z_{2n-1}+K)\right),\\
 &p_5: (f,g)=\left(\sn(z_{2n}+K+iK'), \sn(\gamma_e+K+iK')\right),\\
 &p_6: (f,g)=\left(-\sn(z_{2n}+K+iK'), -\sn(\gamma_e+K+iK')\right),\\
 &p_7: (f,g)=\left(\sn(\gamma_o+K+iK'), \sn(z_{2n-1}+K+iK')\right),\\
 &p_8: (f,g)=\left(-\sn(\gamma_o+K+iK'), -\sn(z_{2n-1}+K+iK')\right),
\end{align}
\end{subequations}
where $K=K(k)$ and $K'=K'(k)$ are complete elliptic integrals, related to periods of ${\rm sn}(z)$. 
For each base point $p_i=(a_j,b_j)$ we can show that the coordinate change
\begin{equation}
 f_j=\cfrac{f-a_j}{g-b_j},\quad
 g_j=g-b_j,
\end{equation}
or
\begin{equation}
 f'_j=f-a_j,\quad
 g'_j=\cfrac{g-b_j}{f-a_j},
\end{equation}
resolve the flow, i.e., the equations contain no other base point on the line $g_j=0$ (or $f'_j=0$).
Therefore the list \eqref{eqn:bps} of base points is completed.
The base points lie on the biquadratic curve given by
\begin{equation}\label{eqn:ellcurve}
A_1(1 + k^2 f^2 g^2) + A_2 fg + A_3 (f^2 + g^2)=0,
\end{equation}
where 
\begin{align}
 A_1=&\sn(\gamma_o+k)\sn(z_{2n-1}+K)(\sn(\gamma_e+K)^2+\sn(z_{2n}+K)^2)\notag\\
 &-\sn(\gamma_e+K)\sn(z_{2n}+K)(\sn(\gamma_o+K)^2+\sn(z_{2n-1}+K)^2),\\
 A_2=&(1+k^2\sn(\gamma_e+K)^2\sn(z_{2n}+K)^2)(\sn(\gamma_o+K)^2+\sn(z_{2n-1}+K)^2)\notag\\
 &-(1+k^2\sn(\gamma_o+K)^2\sn(z_{2n-1}+K)^2)(\sn(\gamma_e+K)^2+\sn(z_{2n}+K)^2),\\
 A_3=&\sn(\gamma_e+K)\sn(z_{2n}+K)(1+k^2\sn(\gamma_o+K)^2\sn(z_{2n-1}+K)^2)\notag\\
 &-\sn(\gamma_o+K)\sn(z_{2n-1}+K)(1+k^2\sn(\gamma_e+K)^2\sn(z_{2n}+K)^2).
\end{align}

In this configuration, no three base points are collinear. 
(This is equivalent to the curve in Equation \eqref{eqn:ellcurve} being an elliptic curve
which we can easily calculate by checking its discriminant.
This is shown by a different method in Section \ref{section:symmetry} by recasting it in terms of a Weierstrass cubic curve.)
The initial value space can be identified with the elliptic surface of type $A_0^{(1)}$ (ell-$A_0^{(1)}$ surface) in Sakai's classification \cite{SakaiH2001:MR1882403}. 
The configuration of base points is invariant under the transformations
\begin{equation}
 (f,g)\to(-f,-g),
\end{equation}
and
\begin{equation}
 (f,g)\to(k^{-1}f^{-1},k^{-1}g^{-1}),
\end{equation}
which are symmetries of the equation.
As the base points act as parameters in construction of Painlev\'e systems through affine Weyl group action, 
we observe that this configuration of pairwise opposite `parameters' would not give the full parameter range expected from the general case.
\subsection{The autonomous case}\label{subsection:autonomous}
In the autonomous case, i.e., System \eqref{eqn:rcgsys} with $z_n=z_0$ (or $\gamma_o=-\gamma_e$), 
or equivalently Equation \eqref{eqn:aut_Q4eqn}, the base points lie on two straight lines: 
\begin{equation}
 \frac{f}{g} = \frac{\sn(z_0+K)}{\sn(\gamma_e+K)}, \quad
 \frac{f}{g} = \frac{\sn(\gamma_e+K)}{\sn(z_0+K)}.
\end{equation}
In this case the rational surface is degenerate and becomes the multiplicative surface of type $A_1^{(1)}$ ($q$-$A_1^{(1)}$ surface). 
(See Appendix \ref{section:qA0} for details.) 
It is interesting to note that the process of deautonomisation through singularity confinement fundamentally changes the initial value space of the equation. 

\subsection{Cremona isometries}\label{subsection:W(E8(1))}
Let $\epsilon: X \to \mathbb{P}^1\times\mathbb{P}^1$ denote blow up of $\mathbb{P}^1\times\mathbb{P}^1$ at the points \eqref{eqn:bps}. 
We show that the group of Cremona isometries associated with this surface $X$ is $W(E_8^{(1)})$ \cite{SakaiH2001:MR1882403}.

The linear equivalence classes of the total transform of the coordinate lines $f$=constant and $g$=constant are denoted by $H_f$ and $H_g$, respectively. 
From \cite{book_HartshorneR1977:MR0463157} we know that the Picard lattice of $X$, denoted by Pic$(X)$, is given by
\begin{equation}\label{eqn:pic_lat}
 {\rm Pic}(X)=\mathbb{Z}H_f\bigoplus\mathbb{Z}H_g\bigoplus_{i=1}^8\mathbb{Z}e_i,
\end{equation}
where $e_i=\epsilon^{-1}(p_i)$, $i=1,\dots,8$, is the total transform of the point of the $i$-th blow up. 
The symmetric bilinear form $(|)$ is defined by the intersection (intersection form) as follows:
\begin{equation}
 (H_f|H_g)=1,\quad
 (H_f|H_f)=(H_g|H_g)=(H_f|e_i)=(H_g|e_i)=0,\quad
 (e_i|e_j)=-\delta_{ij},
\end{equation}
where $1\le i\le 8$, $1\le j\le 8$ are integers. 
The anti-canonical divisor of $X$ is given by
\begin{displaymath}
 \delta=-K_X=2H_f+2H_g-\sum_{i=1}^8e_i.
\end{displaymath}
Because the eight points $\{p_i\}_{i=1}^8$ lie on the biquadratic curve \eqref{eqn:ellcurve}, we identify the surface $X$ as being type ell-$A_0^{(1)}$ in Sakai's list. 

In order to identify the action of the mapping $\varphi_a$ (the time evolution of System \eqref{eqn:rcgsys}) on Pic$(X)$, 
we need to compute its individual actions on $H_f$, $H_g$ and $e_i$, $1\le i\le 8$. 
We provide an example of how this can be done in the following lemma.
\begin{lemma}
$\varphi_a(e_5)=H_f-e_5$.
\end{lemma} 
\begin{proof}
The exceptional line $e_5$ is described by $v_5=0$ in the coordinate chart
\begin{equation}
 (u_5,v_5)=\left(\dfrac{f-\sn(z_{2n}+K+iK')}{g-\sn(\gamma_e+K+iK') },g-\sn(\gamma_e+K+iK')\right).
\end{equation}
This change of variables transforms System \eqref{eqn:rcgsys} to 
\begin{equation}
 (\bar{f},\bar{g})|_{(u_5,v_5)}= (\sn(z_{2n+2}+K+iK'), m(u_5))+{\mathcal O}(v_5),
\end{equation}
where $m$ is a M\"obius transformation. 
Hence $e_5$ is mapped to the curve $\bar{f}=\sn(z_{2n+2}+K+iK')$. 
This curve has the representation $H_f-e_5$ in Pic$(X)$. 
\end{proof}
The action on the other elements of Pic$(X)$ can be calculated in a similar way. 
This provides the following result:
\begin{equation}\label{eqn:picmatrix}
 \varphi_a 
\begin{pmatrix}H_f\\H_g\\e_1\\e_2\\e_3\\e_4\\e_5\\e_6\\e_7\\e_8\end{pmatrix}
 =\begin{pmatrix}
 5 & 2 & -1 & -1 & -1 & -1 & -2 & -2 & -2 & -2\\
 2 & 1 & 0 & 0 & 0 & 0 & -1 & -1 & -1 & -1\\
 2 & 1 & -1 & 0 & 0 & 0 & -1 & -1 & -1 & -1\\
 2 & 1 & 0 & -1 & 0 & 0 & -1 & -1 & -1 & -1\\
 2 & 1 & 0 & 0 & -1 & 0 & -1 & -1 & -1 & -1\\
 2 & 1 & 0 & 0 & 0 & -1 & -1 & -1 & -1 & -1\\
 1 & 0 & 0 & 0 & 0 & 0 & -1 & 0 & 0 & 0\\
 1 & 0 & 0 & 0 & 0 & 0 & 0 & -1 & 0 & 0\\
 1 & 0 & 0 & 0 & 0 & 0 & 0 & 0 & -1 & 0\\
 1 & 0 & 0 & 0 & 0 & 0 & 0 & 0 & 0 & -1
\end{pmatrix}
\begin{pmatrix}H_f\\H_g\\e_1\\e_2\\e_3\\e_4\\e_5\\e_6\\e_7\\e_8\end{pmatrix}.
\end{equation}
\begin{remark}
Algebraic entropy \cite{BV1999:MR1704282} allows one to characterize the integrability of a system by its growth in degree in the initial data. 
The eigenvalues of the matrix corresponding to time evolution in Equation \eqref{eqn:picmatrix} all lie on the unit disk. 
Takenawa showed that if all eigenvalues of the matrix corresponding to time evolution lie on the unit disk, 
then the algebraic entropy of the corresponding system vanishes \cite{TakenawaT2001:MR1877473}.
Thus the entropy of \eqref{eqn:picmatrix} is 0.
\end{remark}

Consider the orthogonal complement $\delta^\bot$. 
We can calculate its root lattice $Q(A_0^{(1)\bot})=\bigoplus_{i=0}^8\mathbb{Z}\alpha_i$
explicitly by searching for elements of Pic$(X)$ that are orthogonal to the anti-canonical divisor $\delta$. 
These lead to the following simple roots
\begin{subequations}\label{eqn:alpha_j}
\begin{align}
 &\alpha_0=e_1-e_2,
 &&\alpha_1=H_f-H_g,\\
 &\alpha_2=H_g-e_1-e_2,
 &&\alpha_3=e_2-e_3,\\
 &\alpha_4=e_3-e_4,
 &&\alpha_5=e_4-e_5,\\
 &\alpha_6=e_5-e_6,
 &&\alpha_7=e_6-e_7,\\
 &\alpha_8=e_7-e_8,
\end{align}
\end{subequations}
where
\begin{equation}
 \delta=3\alpha_0+2\alpha_1+4\alpha_2+6\alpha_3+5\alpha_4+4\alpha_5+3\alpha_6+2\alpha_7+\alpha_8.
\end{equation}
We can show that
\begin{equation}
 (\alpha_i|\alpha_j)
 =\begin{cases}
 -2,& i=j\\
 \ 1, &j=i+1,\, i\neq0,~ {\rm or\ if}\ \ i=3,\, j=0\\
 \ 0, &\text{otherwise}.
\end{cases}
\end{equation}
Representing intersecting $\alpha_i$ and $\alpha_j$ by a line between nodes $i$ and $j$, we obtain the Dynkin diagram of $E_8^{(1)}$ shown in Figure \ref{figure:e8}.
\begin{figure}[hbt]
\begin{center}
\includegraphics[width=0.8\textwidth]{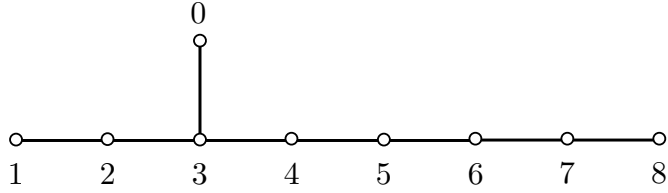}
\caption{Dynkin diagram for the root lattice $\bigoplus_{i=0}^8\mathbb{Z}\alpha_i$.}
\label{figure:e8}
\end{center}
\end{figure}

A Cremona isometry is defined by an automorphism of Pic$(X)$ which preserves 
\begin{description}
\item[(i)]
the intersection form on Pic$(X)$;
\item[(ii)]
the canonical divisor $K_X$;
\item[(iii)]
the semigroup of effective classes of divisors.
\end{description}
It is well known that 
automorphisms of the Dynkin diagram corresponding to the divisors
and reflections for simple roots which orthogonal to all divisors
are Cremona isometries and form (extended) affine Weyl group\cite{DO1988:MR1007155,LooijengaE1981:MR632841,SakaiH2001:MR1882403}.

We define the reflections $s_i$, $i=0, \dots, 8$, across the hyperplane orthogonal to the root $\alpha_i$ by
\begin{equation}
 s_i(v)=v+(v|\alpha_i)\alpha_i 
\end{equation}
for all $v\in {\rm Pic}(X)$. 
It is well known that the reflections $s_i$ form an affine Weyl group $W(E_8^{(1)})$, 
satisfying the following fundamental relations:
\begin{subequations}
\begin{align}
 {s_i}^2&=1,\quad i=0,\dots,8,\\
 (s_is_{i+1})^3&=1,\quad i=1,\dots,7,\\ 
 (s_3s_0)^3&=1,\\
 (s_is_j)^2&=1, \quad {\rm otherwise.}
\end{align}
\end{subequations}
We note here that mapping $1={\rm Id}$ is the identity mapping on ${\rm Pic}(X)$.
(We use this definition throughout this paper.)
We decompose $\varphi_a$ into the generators of $W(E_8^{(1)})$ by $\varphi_a={\varphi_s}^2$ where
\begin{equation}
 \varphi_s=s_5s_4s_3s_0s_6s_5s_4s_3s_0s_7s_6s_5s_4s_3s_0s_8s_7s_6s_5s_4s_3s_1s_2s_0s_3s_4s_0s_3s_2,
\end{equation}
whose left action on the root lattice is given by
\begin{align}
 \varphi_a:&(\alpha_0,\alpha_1,\alpha_2,\alpha_3,\alpha_4,\alpha_5,\alpha_6,\alpha_7,\alpha_8)\notag\\
 &\mapsto
 (-\alpha_0,\alpha_1+\delta,\alpha_0+\alpha_2+2\alpha_3+\alpha_4-\delta,-\alpha_3,-\alpha_4,\notag\\
 &\hspace{2em}-\alpha_0-\alpha_1-2\alpha_2-2\alpha_3-\alpha_4-\alpha_5+\delta,-\alpha_6,-\alpha_7,-\alpha_8),
\end{align}
according to the idea of Lemma 3.11 in \cite{KacVG1990:MR1104219}.
To obtain these results, we used the definitions of $\alpha_i$ \eqref{eqn:alpha_j} and the action of $\varphi_a$ on the Picard lattice \eqref{eqn:picmatrix}. 

Note that the transformation $\varphi_a$ is not a translation on the root lattice, but ${\varphi_a}^2$ is a translation
whose action is given by
\begin{align}
 {\varphi_a}^2:&(\alpha_0,\alpha_1,\alpha_2,\alpha_3,\alpha_4,\alpha_5,\alpha_6,\alpha_7,\alpha_8)\notag\\
 &\mapsto
 (\alpha_0,\alpha_1+2\delta,\alpha_2-2\delta,\alpha_3,\alpha_4,\alpha_5+\delta,\alpha_6,\alpha_7,\alpha_8).
\end{align}
Thus, the system \eqref{eqn:rcgsys} is not obtained by translation in the affine Weyl group $W(E_8^{(1)})$. 
However, the elliptic equation in Sakai's list is obtained by translation in this group. 
(See Section \ref{subsection:SymmetryCremona}.)

\section{Symmetry group for System \eqref{eqn:rcgsys}}\label{section:symmetry}
In this section, we consider the symmetry group for System \eqref{eqn:rcgsys}. 
For this purpose, it is useful to express the base points in terms of Weierstrass' $\wp$ function.
Murata {\it et al.} \cite{MSY2003:MR1958273} obtained actions on these base points. 
We recall it here and provide its iteration under the time evolution of System \eqref{eqn:rcgsys}. 
These are used to find the symmetry group for this system. 
\subsection{Transformation of base points to Weierstrass' $\wp$ form}
The well known transformation of Jacobi to Weierstrass elliptic functions:
\begin{equation}
 \sn(a+b)\sn(a-b)
 =-\sn^2(b)\,\frac{\wp\left(\cfrac{a}{(e_1-e_3)^{1/2}}\right)-e_3-(e_1-e_3)\sn^{-2}(b)}
 {\wp\left(\cfrac{a}{(e_1-e_3)^{1/2}}\right)-e_3-k^2(e_1-e_3)\sn^2(b)},
\end{equation}
motivate our transformation of coordinates $(f,g)$ to $(F,G)$, which are given by
\begin{subequations}
\begin{align}
 &fg=-\sn^2(t)\,\cfrac{F-e_3-(e_1-e_3)\sn^{-2}(t)}{F-e_3-k^2(e_1-e_3)\sn^2(t)},
 \label{eqn:fgtoFG1}\\
 &\cfrac{f}{kg}=-\sn^2\left(t+\frac{\ii K'}{2}\right)\cfrac{G-e_3-(e_1-e_3)\sn^{-2}\left(t+\cfrac{\ii K'}{2}\right)}{G-e_3-k^2(e_1-e_3)\sn^2\left(t+\cfrac{\ii K'}{2}\right)}.
 \label{eqn:fgtoFG2}
\end{align}
\end{subequations}
Here, constants $e_i$, $i=1,2,3$, are the zeroes of the Weierstrass normal cubic $4z^3-g_2z-g_3$,
but these are different to the exceptional lines in the previous section.
This change of coordinates maps the simple base points \eqref{eqn:bps} to double base points
\begin{subequations}\label{eqn:base_p_prime_1}
\begin{align}
 &p'_1=p'_2:(F,G)=
 \left(\wp\left(\frac{v_1}{(e_1-e_3)^{1/2}}\right),
 \wp\left(\frac{v_1}{(e_1-e_3)^{1/2}}+\frac{\ii K'}{2(e_1-e_3)^{1/2}}\right)\right),\\
 &p'_3=p'_4:(F,G)=\left(\wp\left(\frac{v_2}{(e_1-e_3)^{1/2}}\right),
 \wp\left(\frac{v_2}{(e_1-e_3)^{1/2}}+\frac{\ii K'}{2(e_1-e_3)^{1/2}}\right)\right),\\
 &p'_5=p'_6:(F,G)=\left(\wp\left(\frac{v_1+\ii K'}{(e_1-e_3)^{1/2}}\right),
 \wp\left(\frac{v_1+\ii K'}{(e_1-e_3)^{1/2}}+\frac{\ii K'}{2(e_1-e_3)^{1/2}}\right)\right),\\
 &p'_7=p'_8:(F,G)=\left(\wp\left(\frac{v_2+\ii K'}{(e_1-e_3)^{1/2}}\right),
 \wp\left(\frac{v_2+\ii K'}{(e_1-e_3)^{1/2}}+\frac{\ii K'}{2(e_1-e_3)^{1/2}}\right)\right), 
\end{align}
\end{subequations}
where $\ii=\sqrt[]{-1}$.
Note that the correspondence between $\{z_{2n},\gamma_e,\gamma_o\}$ and $\{v_1,v_2,t\}$ are given by
\begin{equation}
 z_{2n}+K=v_1+t,\quad 
 \gamma_e+K=v_1-t,\quad
 \gamma_o+K=v_2+t.
\end{equation}
\subsection{Cremona isometries}\label{subsection:SymmetryCremona}
We consider the algebro-geometric approach to the new coordinate $(F,G)$ and blow up the base points $p'_i$, $i=1,\dots,8$, on the rational surface $X$.
Let $\epsilon: X \to \mathbb{P}^1\times\mathbb{P}^1$ be the blow up of $\mathbb{P}^1\times\mathbb{P}^1$ at the base points given above \eqref{eqn:base_p_prime_1},
$E_i =\epsilon^{-1}(p'_i)$ be the corresponding exceptional divisors, and
$H_0$ and $H_1$ be the linear equivalence classes of total transform of $F$=constant and $G$=constant, respectively.
Therefore, we have the Picard lattice of $X$ and anti-canonical divisor of $X$:
\begin{align}
 &{\rm Pic}(X)=\mathbb{Z}H_0\bigoplus\mathbb{Z}H_1\bigoplus_{i=1}^8\mathbb{Z}E_i,
 \label{sec3:picX}\\
 &\delta=-K_X=2H_0+2H_1-\sum_{i=1}^8E_i,
\end{align}
respectively,
and the intersection form $(|)$ is defined by
\begin{equation}
 (H_0|H_1)=1,\quad
 (H_0|H_0)=(H_1|H_1)=(H_0|E_i)=(H_1|E_i)=0,\quad
 (E_i|E_j)=-\delta_{ij}.
\end{equation}
Furthermore, the orthogonal complement $\delta^\bot=\{\beta_i|i=0,\dots,8\}$ is given by
\begin{subequations}
\begin{align}
 &\beta_1=H_1-H_0,
 &&\beta_2=H_0-E_1-E_2,\\
 &\beta_3=E_2-E_3,
 &&\beta_4=E_3-E_4,\\
 &\beta_5=E_4-E_5,
 &&\beta_6=E_5-E_6,\\
 &\beta_7=E_6-E_7,
 &&\beta_0=E_7-E_8,\\
 &\beta_8=E_1-E_2,
\end{align}
\end{subequations}
where the set of simple roots $\{\beta_0,\dots,\beta_8\}$ corresponds to the Dynkin diagram of $E_8^{(1)}$ (see Figure \ref{fig:dynkin_E8(1)}),
and the anti-canonical divisor is expressed in terms of these by 
\begin{equation}
 \delta=2\beta_1+4\beta_2+6\beta_3+5\beta_4+4\beta_5+3\beta_6+2\beta_7+\beta_0+3\beta_8.
\end{equation}
\begin{figure}[hbt]
\begin{center}
\includegraphics[width=0.8\textwidth]{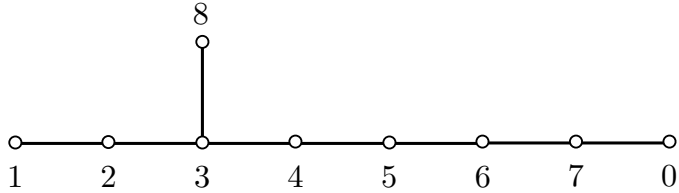}
\caption{Dynkin diagram for the root lattice $\bigoplus_{i=0}^8\mathbb{Z}\beta_i$.}
\label{fig:dynkin_E8(1)}
\end{center}
\end{figure}

We define reflections $w_i,\, i=0,\dots,8$, by 
\begin{equation}
 w_i(v)=v+(v|\beta_i)\beta_i,
\end{equation}
for all $v\in {\rm Pic}(X)$. 
The fundamental relations of affine Weyl group of type $E_8^{(1)}$ hold:
\begin{subequations}
\begin{align}
 {w_i}^2&=1,\quad i=0,\dots,8,\\ 
 (w_iw_{i+1})^3&=1,\quad i=1,\dots,6,\\
 (w_0w_7)^3&=1,\\ 
 (w_3w_8)^3&=1,\\
 (w_iw_j)^2&=1,\quad \text{otherwise}.
\end{align}
\end{subequations}
Left action of the time evolution $\varphi_a$ on Pic$(X)$ is given by
\begin{equation}
 \varphi_a\begin{pmatrix}H_0\\H_1\\E_1\\E_2\\E_3\\E_4\\E_5\\E_6\\E_7\\E_8\end{pmatrix}
 =\begin{pmatrix}
 5&4&-3&-3&-3&-3&-1&-1&-1&-1\\
 4&5&-3&-3&-3&-3&-1&-1&-1&-1\\
 1&1&-1&0&-1&-1&0&0&0&0\\
 1&1&0&-1&-1&-1&0&0&0&0\\
 1&1&-1&-1&-1&0&0&0&0&0\\
 1&1&-1&-1&0&-1&0&0&0&0\\
 3&3&-2&-2&-2&-2&-1&0&-1&-1\\
 3&3&-2&-2&-2&-2&0&-1&-1&-1\\
 3&3&-2&-2&-2&-2&-1&-1&-1&0\\
 3&3&-2&-2&-2&-2&-1&-1&0&-1
\end{pmatrix}
\begin{pmatrix}H_0\\H_1\\E_1\\E_2\\E_3\\E_4\\E_5\\E_6\\E_7\\E_8\end{pmatrix},
\end{equation}
and by following the way in \cite{KacVG1990:MR1104219} the time evolution $\varphi_a$ can be expressed by the elements of $W(E_8^{(1)})$ as
\begin{subequations}
\begin{align}
 \varphi_a=&{\varphi_s}^2,\\
 \varphi_s=&
   w_2w_3w_8w_4w_3w_2w_1w_2w_3w_4
   w_5w_6w_7w_0w_8w_3w_4w_5w_6w_7\notag\\
 &w_0w_2w_3w_4w_5w_6w_8w_3w_4w_5
   w_6.
   \label{eqn:defn_phi_s}
\end{align}
\end{subequations}
Note that Sakai's elliptic discrete Painlev\'e equation is given by\cite{MSY2003:MR1958273}
\begin{align}
 T=&w_1w_2w_3w_8w_4w_3w_2w_5w_4w_3w_8w_6w_5w_4
 w_3w_2w_7w_6w_5w_4w_3w_8\notag\\
 &w_0w_7w_6w_5w_4w_3w_2
 w_1w_2w_3w_4w_5w_6w_7w_0
 w_8w_3w_4w_5w_6w_7w_2\notag\\
 &w_3w_4w_5w_6w_8w_3w_4w_5w_2w_3w_4w_8w_3w_2.
\end{align}
\subsection{Birational representation of $W(E_8^{(1)})$}
Now we are ready to use the representation of $W(E_8^{(1)})$ given in \cite{MSY2003:MR1958273}.
Let us consider the general setting of base points \eqref{eqn:base_p_prime_1} by
\begin{equation}\label{eqn:base_p_prime_2}
 p'_i:(F,G)=\Big(\wp(u_i+b),\wp(u_i-b)\Big),\quad i=1,\dots,8.
\end{equation}
The left action of $w_i$ is given as
\allowdisplaybreaks{
\begin{subequations}\label{eqn:actions_wi_all}
\begin{align}
 &w_1:\,(b,F,G)\mapsto (-b,G,F),
 &&w_i:\,(u_{i-1},u_i)\mapsto(u_i,u_{i-1}), ~i=3,\dots,7,\label{eqn:action_w_i}\\
 &w_0:\,(u_7,u_8)\mapsto(u_8,u_7),
 &&w_8:\,(u_1,u_2)\mapsto(u_2,u_1),
\end{align}
\begin{align}
 &w_2(u_1)=u_1-\cfrac{3(2b+u_1+u_2)}{4},\quad
 &&w_2(u_2)=u_2-\cfrac{3(2b+u_1+u_2)}{4},\\
 &w_2(u_3)=u_3+\cfrac{2b+u_1+u_2}{4},\quad
 &&w_2(u_4)=u_4+\cfrac{2b+u_1+u_2}{4},\\
 &w_2(u_5)=u_5+\cfrac{2b+u_1+u_2}{4},\quad
 &&w_2(u_6)=u_6+\cfrac{2b+u_1+u_2}{4},\\
 &w_2(u_7)=u_7+\cfrac{2b+u_1+u_2}{4},\quad
 &&w_2(u_8)=u_8+\cfrac{2b+u_1+u_2}{4},\\
 &w_2(b)=b-\cfrac{2b+u_1+u_2}{4},
\end{align}
\begin{align}
 \cfrac{w_2(G)-\wp\left(2b-\cfrac{u_1-u_2}{2}\right)}{w_2(G)-\wp\left(2b+\cfrac{u_1-u_2}{2}\right)}
 =&\cfrac{\left(F-\wp(u_2+b)\right)\left(G-\wp(b-u_1)\right)\left(\wp(2b)-\wp(b-u_2)\right)}
  {\left(F-\wp(u_1+b)\right)\left(G-\wp(b-u_2)\right)\left(\wp(2b)-\wp(b-u_1)\right)}\notag\\
 &\times\cfrac{\wp\left(b-\cfrac{u_1+u_2}{2}\right)-\wp\left(2b-\cfrac{u_1-u_2}{2}\right)}
 {\wp\left(b-\cfrac{u_1+u_2}{2}\right)-\wp\left(2b+\cfrac{u_1-u_2}{2}\right)}.
\end{align}
\end{subequations}
}
Therefore, the time evolution $\varphi_a$ acts on the parameters as
\begin{align}
 \varphi_a:&(u_1,u_2,u_3,u_4,u_5,u_6,u_7,u_8)\notag\\
 &\mapsto(u_2+\lambda,u_1+\lambda,u_4+\lambda,u_3+\lambda,u_6-\lambda,u_5-\lambda,u_8-\lambda,u_7-\lambda),
\end{align}
where $\lambda=\frac{1}{2}\sum_{i=1}^8u_i$ is invariant under the action of $W(E_8^{(1)})$.
Since $\varphi_a$ is not a translation on the parameter space
\begin{equation}
 (u_1,u_2,u_3,u_4,u_5,u_6,u_7,u_8,b),
\end{equation}
we cannot regard it as a time evolution of difference equation.
However, by considering $\varphi_a$ on the subspace of parameters given by the specialization
\begin{equation}\label{eqn:subspace}
 u_1=u_2,\quad u_3=u_4,\quad u_5=u_6,\quad u_7=u_8,
\end{equation}
obtained by comparing \eqref{eqn:base_p_prime_1} and \eqref{eqn:base_p_prime_2},
we find that $\varphi_a$ describes the translational motion on the subspace
\begin{equation}
 \varphi_a:(u_1,u_3,u_5,u_7)
 \to(u_1+\lambda,u_3+\lambda,u_5-\lambda,u_7-\lambda).
\end{equation}
The process of deriving discrete dynamical systems of Painlev\'e type from elements of affine Weyl groups that are of infinite order (but that are not necessarily translations) 
by taking a projection on an appropriate subspace of parameters
is called a {\it projective reduction} \cite{KNT2011:MR2773334}.
\subsection{Symmetry group for System \eqref{eqn:rcgsys}}
As seen above, we have a group $W(E_8^{(1)})$ acting on the parameters $u_i$, $i=1,\dots,8$, and $b$.
Now we introduce $\Omega\subset W(E_8^{(1)})$ as the subgroup of elements that preserve the relation \eqref{eqn:subspace},
that is,
transformations $\rho\in W(E_8^{(1)})$ which satisfy
\begin{equation}\label{eqn:symmetry_condition1}
 \rho(u_1)=\rho(u_2),\quad \rho(u_3)=\rho(u_4),\quad \rho(u_5)=\rho(u_6),\quad \rho(u_7)=\rho(u_8),
\end{equation}
under the condition \eqref{eqn:subspace}.
The symmetry group for System \eqref{eqn:rcgsys} is then defined by
\begin{equation}
 \Omega
 =\left\{\left.\rho\in W(E_8^{(1)})\right|
 \text{$\rho$ satisfies \eqref{eqn:symmetry_condition1} under the condition \eqref{eqn:subspace}}\right\}.
\end{equation}
\begin{lemma}\label{lemma:symmetry}
The symmetry group for System \eqref{eqn:rcgsys} is given by
\begin{equation}
 \Omega=\langle w_0,w_1,w_2,r_1,w_4,r_2,w_6,r_3,w_8\rangle,
\end{equation}
where
\begin{equation}\label{eqn:def_r123}
 r_1=w_3w_4w_8w_3,\quad
 r_2=w_5w_4w_6w_5,\quad
 r_3=w_7w_6w_0w_7.
\end{equation}
\end{lemma}
The proof of Lemma \ref{lemma:symmetry} is given in Appendix \ref{section:proof_symmetry}.

We note that transformations $\varphi_a$ and $\varphi_s$ are elements of $\Omega$ since
\begin{equation}\label{eqn:wr_phi}
 \varphi_a={\varphi_s}^2,\quad
 \varphi_s=w_2r_1w_2w_1w_2r_1w_2r_2r_3r_1r_2w_0w_6.
\end{equation}
We now quotient out the subgroup of elements which are identity transformations on the parameters and variables under the conditions \eqref{eqn:subspace}.
This includes, in particular, the generators $w_i$, $i=0,4,6,8$, so they are removed from symmetry groups for System \eqref{eqn:rcgsys}.
The symmetry group becomes
\begin{equation}
 \langle w_1,w_2,r_1,r_2,r_3\rangle.
\end{equation}
Therefore, we obtain the following theorem.
\begin{thm}
The symmetry group for System \eqref{eqn:rcgsys} gives a representation of an affine Weyl group of type $F_4^{(1)}$,
\begin{equation}
\langle w_1,w_2,r_1,r_2,r_3\rangle=W(F_4^{(1)})
\end{equation}
\end{thm}
\begin{proof}
The generators $w_i$, $i=1,2$, and $r_i$, $i=1,2,3$, satisfy the following relations
\begin{subequations}
\begin{align}
 &{w_1}^2={w_2}^2={r_1}^2={r_2}^2={r_3}^2=1,\\
 &(w_1w_2)^3=(w_1r_1)^2=(w_1r_2)^2=(w_1r_3)^2=1,\\
 &(w_2r_1)^4=(w_2r_2)^2=(w_2r_3)^2=1,\quad
 (r_1r_2)^3=(r_1r_3)^2=1,\quad
 (r_2r_3)^3=1,
\end{align}
\end{subequations}
and so $\langle w_1,w_2,r_1,r_2,r_3\rangle$ gives a representation of an affine Weyl group of type $F_4^{(1)}$ (see Figure \ref{fig:dynkin_F4(1)}). 
On the parameters this representation is linear (cf. \eqref{eqn:actions_wi_all}), so that its faithfulness can be checked directly.
\end{proof}

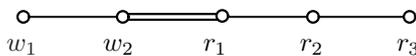
\begin{figure}[hbt]
\begin{center}
{\unitlength 0.1in%
\begin{picture}(21.9500,1.2500)(5.2000,-7.9000)%
%
\special{pn 13}%
\special{pa 700 700}%
\special{pa 1200 700}%
\special{fp}%
\put(6.8000,-8.5500){\makebox(0,0){$w_1$}}%
%
\special{pn 13}%
\special{pa 1700 700}%
\special{pa 2200 700}%
\special{fp}%
%
\special{pn 13}%
\special{pa 2200 700}%
\special{pa 2700 700}%
\special{fp}%
%
\special{sh 0}%
\special{ia 685 700 30 30 0.0000000 6.2831853}%
\special{pn 13}%
\special{ar 685 700 30 30 0.0000000 6.2831853}%
%
\special{sh 0}%
\special{ia 2685 700 30 30 0.0000000 6.2831853}%
\special{pn 13}%
\special{ar 2685 700 30 30 0.0000000 6.2831853}%
%
\special{sh 0}%
\special{ia 2185 700 30 30 0.0000000 6.2831853}%
\special{pn 13}%
\special{ar 2185 700 30 30 0.0000000 6.2831853}%
%
\special{pn 13}%
\special{pa 1210 685}%
\special{pa 1710 685}%
\special{fp}%
%
\special{pn 13}%
\special{pa 1215 715}%
\special{pa 1715 715}%
\special{fp}%
%
\special{sh 0}%
\special{ia 1200 700 30 30 0.0000000 6.2831853}%
\special{pn 13}%
\special{ar 1200 700 30 30 0.0000000 6.2831853}%
%
\special{sh 0}%
\special{ia 1715 695 30 30 0.0000000 6.2831853}%
\special{pn 13}%
\special{ar 1715 695 30 30 0.0000000 6.2831853}%
\put(11.8000,-8.5500){\makebox(0,0){$w_2$}}%
\put(16.8000,-8.5500){\makebox(0,0){$r_1$}}%
\put(21.8000,-8.5500){\makebox(0,0){$r_2$}}%
\put(26.8000,-8.5500){\makebox(0,0){$r_3$}}%
\end{picture}}%
\caption{Dynkin diagram for the affine Weyl group $W(F_4^{(1)})=\langle w_1,w_2,r_1,r_2,r_3\rangle$.}
\label{fig:dynkin_F4(1)}
\end{center}
\end{figure}

\section{Relation with Q4}\label{section:solution}
\subsection{Alternative characterisation}\label{sec:AC}
Due to its construction via periodic reduction followed by integrable deautonomization, equation \eqref{eqn:rcg1} was not originally viewed as defining some class of solutions for Q4. Here we give an alternative construction that shows how it does in fact constitute a reduction in this more usual sense.

The reduction itself is slightly different to the better-studied periodic and quasi-periodic reductions.
It comes from a constraint local to a specified path along the lattice. 
Due to locality, this kind of reduction is admissible regardless of whether the equation is autonomous.
An additional condition on the lattice parameters is therefore required in order to obtain equation \eqref{eqn:rcg1}.

The quad-equation Q4, in Jacobi form with fixed modulus $k$, is defined by the two-parameter family of four-variable polynomials
\begin{multline}\label{eqn:Q4}
Q_{a,b}(w,x,y,z):=\sn(a)(wx+yz)-\sn(b)(wy+xz)\\
-\sn(a-b)[xy+wz-k\sn(a)\sn(b)(1+wxyz)].
\end{multline}
Here this quad-equation will be considered on the most regular planar quad-graph, $\mathbb{Z}^2$, that is
\begin{equation}\label{eqn:abs}
Q_{\alpha_n,\beta_m}(v_{n,m},v_{n+1,m},v_{n,m+1},v_{n+1,m+1})=0,
\end{equation}
where $n,m\in\mathbb{Z}$, variables $v_{n,m}$ are assigned to vertices, and $\alpha_n$ and $\beta_m$ are parameters on characteristics. (In general, if two edges are on opposite sides of a quad, then,  following \cite{AdVeQ}, we say they are on the same {\it characteristic}. The characteristics define a partition of the set of edges of a quad-graph, for $\mathbb{Z}^2$ they are just horizontal and vertical strips.)

The constraint defining the first part of the reduction from system \eqref{eqn:abs} is
\begin{equation}\label{eqn:ridge}
v_{n,n+1}=v_{n+1,n}, \quad v_{n,n}=v_{n-1,n+1}, \quad n\in\mathbb{Z}.
\end{equation}
It is local to a path along a diagonal of $\mathbb{Z}^2$.
It is easily verified that suitable initial data for equation \eqref{eqn:abs} with constraint \eqref{eqn:ridge} can be the two variables $v_{0,0}$ and $v_{0,1}$, in particular, this reduction leaves two degrees of freedom.

The resulting dynamics are not integrable (in the zero-entropy sense) in general.
However, the special choice of the parameters,
\begin{equation}\label{eqn:parameter_constraint}
\alpha_n=\alpha_0+n\phi, \quad \beta_m=\beta_0-m\phi,
\end{equation}
where $\alpha_0,\beta_0,\phi\in\mathbb{C}$ are free constants, results in quadratic degree growth (observed experimentally) along any direction of the lattice.
The equation \eqref{eqn:rcg1} is recovered by restricting attention to the path of the constraint \eqref{eqn:ridge}, specifically, identifying the variables of \eqref{eqn:rcg1} via the relations
\begin{equation}\label{eqn:con1}
(u_{2n},u_{2n+1})=\left(\frac{v_{n,n}}{\sqrt{k}},\frac{v_{n,n+1}}{\sqrt{k}}\right),\quad n\in\mathbb{Z}, 
\end{equation}
and the parameters of \eqref{eqn:rcg1} via the relations
\begin{equation}\label{eqn:con2}
\quad z_0=(\alpha_0-\beta_0)/2,\ \gamma_e=(\alpha_0+\beta_0)/2,\ \gamma_o=(\phi-\alpha_0-\beta_0)/2.
\end{equation}

In summary we formulate the following:
\begin{prop}
The lattice system \eqref{eqn:abs} under the constraint on variables \eqref{eqn:ridge} and choice of parameters \eqref{eqn:parameter_constraint} defines a system with two degrees of freedom, which, on the subset of variables $v_{n,n},v_{n,n+1},n\in\mathbb{Z}$, coincides with \eqref{eqn:rcg1} modulo identifications \eqref{eqn:con1} and \eqref{eqn:con2}.
\end{prop}

\begin{remark}
If it is chosen that $\phi=0$ in \eqref{eqn:parameter_constraint}, it follows that $v_{n,m+1}=v_{n+1,m}$ throughout the lattice as a consequence of \eqref{eqn:ridge}. In other words, the autonomous case of this reduction coincides with the $(1,1)$-periodic reduction. In that case, the system has an invariant that was found in \cite{JGTR2006:MR2271126} and its general solution is in terms of elliptic functions \cite{AN2008:MR2452416}.
\end{remark}

\begin{remark}
From the point of view of exploring natural generalisations of this reduction, it is most desirable to find a self-contained explanation for the parameter constraint \eqref{eqn:parameter_constraint}, for instance in terms of the lattice-geometry of Q4 (its multidimesionality or some extension thereof \cite{INIFano}), instead of imposing integrability from the outside.
\end{remark}

\subsection{Singular solutions}
The previous subsection identifies \eqref{eqn:rcg1} with a particular class of solutions of Q4. This allows us to apply elements of the theory of Q4 to this equation. In particular we exploit here some admissable singularity patterns \cite{ATK2011} in order to construct special solutions compatible with the reduction, but it would also be interesting to consider the action of the B\"acklund transformation on this reduction, or to investigate it in the context of the associated linear problem for Q4 \cite{Nij2002}.

Singularities in solutions of Q4 have been fundamental to characterisation of its defining polynomial \cite{ABS2009:MR2503862}. The singularities are associated naturally with lattice edges (through vanishing of edge biquadratics).
How singularities may extend globally was studied in \cite{ATK2011}.
Based on those criteria, it is straightforward to identify an admissible singularity pattern for the system \eqref{eqn:abs} that is compatible with the constraint \eqref{eqn:ridge}:
\begin{equation}\label{eqn:singular_solution1}
\begin{split}
v_{n,n}&=\sqrt{k}\sn\left(\xi_0+\sum_{j=0}^{j=n-1}\alpha_j-\sum_{j=0}^{j=n-1}\beta_j\right),\\
v_{n,n+1}&=\sqrt{k}\sn\left(\xi_0+\sum_{j=0}^{j=n-1}\alpha_j-\sum_{j=0}^{j=n}\beta_j\right),
\end{split}
\end{equation}
were $\xi_0$ is a free constant. The singular edges here connect consecutive pairs $(v_{n,n},v_{n,n+1})$ and $(v_{n,n+1},v_{n+1,n+1})$, $n\in\mathbb{Z}$.

Assuming \eqref{eqn:singular_solution1}, the problem to determine $v_{n,m}$ for $m\not\in\{n,n+1\}$ that satisfies \eqref{eqn:abs} under the constraint \eqref{eqn:ridge}, i.e., to extend the solution \eqref{eqn:singular_solution1} outside of the singular region, requires some integration. 
Interestingly, it turns out that the imposed singular region renders the system integrable, in the sense of vanishing entropy orthogonal to the path of the singularity, without the assumption \eqref{eqn:parameter_constraint}.

The parameter choice \eqref{eqn:parameter_constraint} brings \eqref{eqn:singular_solution1} to the form
\begin{equation}\label{eqn:singular_solution2}
\begin{split}
v_{n,n}&=\sqrt{k}\sn(\xi_0+n(\alpha_0-\beta_0)+n(n-1)\phi),\\
v_{n,n+1}&=\sqrt{k}\sn(\xi_0-\beta_0+n(\alpha_0-\beta_0)+n^2\phi).
\end{split}
\end{equation}
This therefore defines a particular solution of equation (\ref{eqn:rcg1}) modulo identifications (\ref{eqn:con1}) and (\ref{eqn:con2}), which can also be verified directly by substitution.
Direct comparison shows this singular solution shares important features with the known trivial solution of the elliptic Painlev\'e equation \cite{MSY2003:MR1958273}, which is given by 
\begin{equation}\label{eqn:elpsoln}
 f=\wp (q_0+2t^2/\lambda+t),\quad
 g=\wp (-q_0-2t^2/\lambda+t).
\end{equation} 
Here $t$ is the independent variable which additively shifts in increments of $\lambda$, and the parameter $q_0$ is given by the initial condition. 
Both \eqref{eqn:singular_solution2} and \eqref{eqn:elpsoln} have quadratic dependence on the independent variable, the notable differences are probably due to the fact that they do not correspond to the same component in the Weyl group of $E_8^{(1)}$.

\subsection{Similar kinds of reduction}
Here we report some further investigations of the general idea of reduction-by-a-local-constraint, looking at another example of a reduction from Q4 on $\mathbb{Z}^2$.

Thus, consider again the system \eqref{eqn:abs}, with the constraint
\begin{equation}\label{eqn:odd}
v_{n-1,n+1}=v_{n,n}=v_{n+1,n-1},\quad n\in\mathbb{Z},
\end{equation}
which is again local to the $\mathbb{Z}^2$ diagonal.
It is straightforward to verify that suitable initial data is $v_{1,0},v_{0,0},v_{0,1}$, so that, in particular, there remain three degrees of freedom.

In general this reduction is not integrable, however, choosing parameters in the form
\begin{equation}\label{eqn:constraint3}
\alpha_n=\alpha_0+n(\alpha_1-\alpha_0), \quad \beta_m=\beta_0+m(\beta_1-\beta_0),
\end{equation}
where $\alpha_0,\alpha_1,\beta_0,\beta_1\in\mathbb{C}$ are free constants, is sufficient for integrability.
This is observed experimentally. Iterates on the subset of vertices where $n+m$ is odd have cubic degree growth in the initial data, whilst on the sites where $n+m$ is even, the growth is quadratic.

The interesting degree-growth pattern can be understood by bearing in mind the relation between quad-equation \eqref{eqn:Q4} and its associated Toda-type system, the theory of which has been developed in \cite{AdlSur}.
Specifically, a consequence of \eqref{eqn:abs} is an equation relating, for any $n,m\in\mathbb{Z}$, the five variables 
\begin{equation}\label{eqn:five}
v_{n,m},\ v_{n+1,m+1},\ v_{n-1,m+1},\ v_{n+1,m-1},\ v_{n-1,m-1}.
\end{equation}
An expression for it, obtained directly from the three-leg form of the quad-equation \cite{ABS2003:MR1962121}, is as follows:
\begin{multline}\label{eqn:todatype1}
F(\xi_{n,m},\xi_{n+1,m+1},\alpha_{n+1}-\beta_{m+1})
F(\xi_{n,m},\xi_{n-1,m-1},\alpha_{n}-\beta_{m})=\\
F(\xi_{n,m},\xi_{n-1,m+1},\alpha_{n}-\beta_{m+1})
F(\xi_{n,m},\xi_{n+1,m-1},\alpha_{n+1}-\beta_{m}),
\end{multline}
where
\begin{equation}
v_{n,m}=\sqrt{k}\sn(\xi_{n,m}), \quad F(\xi,\zeta,\gamma) := \frac{\sn(\zeta)-\sn(\xi+\gamma)}{\sn(\zeta)-\sn(\xi-\gamma)}\frac{\Theta(\xi+\gamma)}{\Theta(\xi-\gamma)}.
\end{equation}
The way \eqref{eqn:todatype1} is written is more conceptual, but obscures the rationality. It is, in fact, equivalent to vanishing of a polynomial in the five variables \eqref{eqn:five}, which is degree two in $v_{n,m}$, and degree one in each of the four remaining variables.

Supposing $\xi_{n,m}$ with $n+m$ even satisfy \eqref{eqn:todatype1}, then, say, $v_{1,0}$ can be chosen freely, and the equation \eqref{eqn:abs} can be used to consistently determine the solution throughout the remainder of the lattice, i.e., the $v_{n,m}$ where $n+m$ is odd.
In this way, solution of the original system \eqref{eqn:abs} is recovered from any solution of \eqref{eqn:todatype1}, demonstrating their equivalence on the level of solutions.
That the iteration can be decoupled in this way gives the mechanism leading to cubic degree growth for variables on the odd lattice sites.

Combined with \eqref{eqn:odd}, the equation \eqref{eqn:todatype1} yields, for each $n\in\mathbb{Z}$, a relation between variables, $v_{n-1,n-1},v_{n,n}$ and $v_{n+1,n+1}$ (which are on the even lattice sites). 
Manipulating \eqref{eqn:todatype1} to obtain its polynomial form, imposing \eqref{eqn:odd}, and substituting also the prescribed parameters \eqref{eqn:constraint3}, the relation is found to be
\begin{multline}\label{eqn:todatype2}
\left[k(1-k\omega)\sn(2\gamma)\sn(2a)\sn(2b)-\omega/\sn(\gamma-a-b)-\omega/\sn(\gamma+a+b)\right](u-\utilde{u}u^3\tilde{u})\\
+\left[\omega/\sn(\gamma-a+b)\right](\utilde{u}-\tilde{u}u^4)
+\left[\omega/\sn(\gamma+a-b)\right](\tilde{u}-\utilde{u}u^4)\\
+\left[\sn(\gamma+a-b)-\sn(\gamma-a+b)-(1-k\omega)(\sn(2a)-\sn(2b))\right](\utilde{u}u^2-\tilde{u}u^2)\\
-\left[\sn(\gamma+a+b)+\sn(\gamma-a-b)-(1-k\omega)\sn(2\gamma)\right](\utilde{u}u\tilde{u}-u^3)=0.
\end{multline}
Here 
\begin{multline}\label{eqn:TGvars}
\utilde{u}=v_{n-1,n-1}, \ u=v_{n,n}, \ \tilde{u}=v_{n+1,n+1}, \ \gamma=\gamma_0+2(a-b)n,\\
\omega=k\sn(\gamma+a+b)\sn(\gamma-a+b)\sn(\gamma+a-b)\sn(\gamma-a-b),
\end{multline}
and $\gamma_0$, $a$ and $b$ are free constants that replace $\alpha_0,\alpha_1,\beta_0$ and $\beta_1$,
\begin{equation}\label{eqn:TGparams}
a=\frac{1}{2}\left(\alpha_1-\alpha_0\right), \quad b=\frac{1}{2}\left(\beta_1-\beta_0\right), \quad \gamma_0=\frac{1}{2}\left(\alpha_0+\alpha_1-\beta_0-\beta_1\right).
\end{equation}
Like \eqref{eqn:rcg1}, the equation \eqref{eqn:todatype2} is an integrable, second-order, non-autonomous discrete equation that defines a natural class of solutions of Q4, and which merits investigation as an interesting sub-case of the elliptic Painlev\'e equation.

For the equation \eqref{eqn:todatype2} it is straightforward to write down a trivial solution
\begin{equation}\label{eqn:singsol2}
u=\sqrt{k}\sn(\xi_0+\gamma_0n+(a-b)n^2),
\end{equation}
where $\xi_0$ may be chosen freely, as a consequence of a singular solution of \eqref{eqn:abs} that is compatible with constraint \eqref{eqn:odd}. This kind of singular solution is easily verified using \eqref{eqn:todatype1}.
\begin{remark}
The natural trigonometric and rational limits of equations \eqref{eqn:rcg1} and \eqref{eqn:todatype2} have general solution that is expressible in terms of their respective function class.
It is possible that equations \eqref{eqn:rcg1} and \eqref{eqn:todatype2} have a general solution that is expressible in terms of elliptic functions.
\end{remark}

\section{Conclusion and Discussion}\label{section:conc}
In this paper we have studied equation \eqref{eqn:rcg1} as an integrable, second-order, discrete equation, with elliptic non-autonomous term.
It was obtained originally by Ramani {\it et al.} \cite{RCG2009:MR2525848}, as an integrable deautonomization of the $(1,1)$ periodic reduction of Q4.

We found the initial value space to be that of ell-$A_0^{(1)}$ in Sakai's list \cite{SakaiH2001:MR1882403}, placing the equation within the geometric framework of the discrete Painlev\'e equations.
Sakai provided the beautiful elliptic discrete Painlev\'e equation as the most general case; constructing it as a translation arising from $W(E_8^{(1)})$. 
Our results show that equation \eqref{eqn:rcg1} has symmetry which is a proper sub-group of $W(E_8^{(1)})$.
It corresponds to a `projective reduction'; a parameter sub-case of the general construction in which the \lq square-root\rq\ of a translation yields the difference equation.

We have established a more direct contact with Adler's discrete analog of the Krichever-Novikov equation \cite{AdlerVE1998:MR1601866}, known as Q4, showing that equation \eqref{eqn:rcg1} defines a particular class of solutions.
The notion we introduce is a constraint local to the diagonal of the $\mathbb{Z}^2$ lattice \eqref{eqn:ridge}, which yields a two-degree-of-freedom system.
In general, this system has non-vanishing entropy. 
It is then a special prescription of the lattice parameters \eqref{eqn:parameter_constraint} that picks out a zero-entropy case, giving rise to \eqref{eqn:rcg1}.
This view allows the application of elements of the theory of Q4 to \eqref{eqn:rcg1}. For instance, from a singular solution of Q4 we obtain directly a one-parameter solution for \eqref{eqn:rcg1}. This is identifiable with the known trivial solution of the Painlev\'e equation.
Also we provide \eqref{eqn:todatype2}, which is a new equation obtained by a similar, but more symmetric reduction.
It is in the same class as \eqref{eqn:rcg1} and has similar features, so it provides an interesting candidate for further study.

\section*{Funding}
This work has been supported by 
an Australian Postdoctoral Fellowship DP110104151, 
an Australian Postgraduate Award, 
an Australian Laureate Fellowship FL120100094 and grant DP130100967 from the Australian Research Council, 
and a JSPS Grant-in-Aid for Scientific Research 22$\cdot$4366. 
\section*{Acknowledgement}
The authors would like to sincerely thank Prof. M. Noumi for helpful discussions. 
We would also thank Profs T. Masuda, T. Takenawa, T. Tsuda and Y. Yamada 
and Drs P. McNamara and Y. Shi
for their valuable insights and encouragement.
\appendix
\section{System associated with $A_1^{(1)}$-surface}\label{section:qA0}
In \cite{CarsteaAS2013:MR3087958}, Carstea showed that the rational surface of Equation \eqref{eqn:aut_Q4eqn} is of $A_1^{(1)}$-type.
In this section, we start with the setting where eight base points $p_i$, $i=1,\dots,8$,
arranged in pairs with four lying on each of lines $L_1$ and $L_2$ given by
\begin{equation}
 L_1:~f-v_1g=0,\quad
 L_2:~f-v_2g=0,
\end{equation}
where $v_1v_2=1$.
This is a generalization of the setting for Equation \eqref{eqn:aut_Q4eqn}.
(See Section \ref{subsection:autonomous}.) 
We show that another non-autonomous form of Equation \eqref{eqn:aut_Q4eqn} emerges 
as a dynamical system in this generalized setting.

Let
\begin{subequations}\label{eqn:appendix_bps}
\begin{align}
 &p_1: (f,g)=(-v_1u_1,-u_1),
 &&p_2: (f,g)=(-v_1u_2,-u_2),\\
 &p_3: (f,g)=(-v_1u_3,-u_3),
 &&p_4: (f,g)=(-v_1u_4,-u_4),\\
 &p_5: (f,g)=(-v_2u_5,-u_5), 
 &&p_6: (f,g)=(-v_2u_6,-u_6),\\
 &p_7: (f,g)=(-v_2u_7,-u_7),
 &&p_8: (f,g)=(-v_2u_8,-u_8),
\end{align}
\end{subequations}
where points $p_i$, $i=1,...,4$, are on the line $L_1$ 
and points $p_j$, $j=5,...,8$, are on the line $L_2$.

Let $\epsilon: X \to \mathbb{P}^1\times\mathbb{P}^1$ denotes blow up of $\mathbb{P}^1\times\mathbb{P}^1$ at the points \eqref{eqn:appendix_bps}.
Pic$(X)$ is given by
\begin{equation}
 {\rm Pic}(X)=\mathbb{Z}H_f\bigoplus\mathbb{Z}H_g\bigoplus_{i=1}^8\mathbb{Z}e_i,
\end{equation}
where $H_f$, $H_g$, $e_i=\epsilon^{-1}(p_i)$ are the linear equivalence class of total transform of $f$=constant, 
that of $g$=constant and the total transform of the point of the $i$-th blow up, respectively.
The intersection form $(|)$ is defined by 
\begin{equation}
 (H_f|H_g)=1,\quad
 (H_f|H_f)=(H_g|H_g)=(H_f|e_i)=(H_g|e_i)=0,\quad
 (e_i|e_j)=-\delta_{ij},
\end{equation}
where $1\le i\le 8$, $1\le j\le 8$ are integers and
the anti-canonical divisor of $X$ is uniquely decomposed into the prime divisors:
\begin{equation}
 \delta=-K_X=2H_f+2H_g-\sum_{i=1}^8 e_i=D_1+D_2,
\end{equation}
where
\begin{equation}
 D_1=H_f+H_g-e_1-e_2-e_3-e_4,\quad
 D_2=H_f+H_g-e_5-e_6-e_7-e_8.
\end{equation}
Thus, we identify the surface $X$ as being type $A_1^{(1)}$ in Sakai's list.
Furthermore, the orthogonal complement is $\delta^\bot=\{\alpha_i|i=0,\dots,7\}$ where
\begin{subequations}
\begin{align}
 &\alpha_0=H_f-H_g,
 &&\alpha_1=-e_1+e_2,\\
 &\alpha_2=-e_2+e_3,
 &&\alpha_3=-e_3+e_4,\\
 &\alpha_4=H_g-e_4-e_5,
 &&\alpha_5=e_5-e_6,\\
 &\alpha_6=e_6-e_7,
 &&\alpha_7=e_7-e_8,
\end{align}
\end{subequations}
and satisfy
\begin{equation}
 \delta=\alpha_1+2\alpha_2+3\alpha_3+4\alpha_4+3\alpha_5+2\alpha_6+\alpha_7+2\alpha_0.
\end{equation}
The root lattice $\bigoplus_{i=0}^7\mathbb{Z}\alpha_i$ corresponds to the Dynkin diagram of $E_7^{(1)}$ since
\begin{equation}
 (\alpha_i|\alpha_j)
 =\begin{cases}
 -2,& i=j\\
 \ 1, &i=j+1,\, i=1,\dots,6,~ {\rm or\ if}\ \ i=4,\, j=0\\
 \ 0, &\text{otherwise}.
\end{cases}
\end{equation}
\begin{figure}[hbt]
\begin{center}
\hspace*{5em}\includegraphics[width=0.8\textwidth]{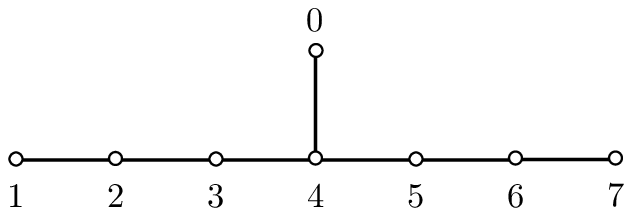}
\caption{Dynkin diagram for the root lattice $\bigoplus_{i=0}^7\mathbb{Z}\alpha_i$.}
\end{center}
\end{figure}

We define the reflections $s_i$, $i=0, \dots, 7$, across the hyperplane orthogonal to the root $\alpha_i$ by
\begin{equation}
 s_i(\lambda)=\lambda+(\lambda|\alpha_i)\alpha_i 
\end{equation}
for all $\lambda\in {\rm Pic}(X)$ and the diagram automorphisms Aut$(A_1^{(1)})=\langle\sigma\rangle$ by
\begin{equation}
 \sigma:(e_1,e_2,e_3,e_4,e_5,e_6,e_7,e_8)
 \mapsto(e_8,e_7,e_6,e_5,e_4,e_3,e_2,e_1).
\end{equation}
We can easily verify that the actions of $\widetilde{W}(E_7^{(1)})=\langle s_0,\dots,s_7,\sigma\rangle$
on Pic$(X)$ satisfy the fundamental relations
of the extended affine Weyl group of type $E_7^{(1)}$:
\begin{subequations}
\begin{align}
 {s_i}^2&=1,\quad i=0,\dots,7,\\
 (s_is_{i+1})^3&=1,\quad i=1,\dots,6,\\
 (s_4s_0)^3&=1,\\
 (s_is_j)^2&=1,\quad\text{otherwise},\\
 \sigma^2&=1,\\
 \sigma s_{(1,2,3,5,6,7)}&=s_{(7,6,5,3,2,1)}\sigma.
\end{align}
\end{subequations}
In Sakai's classification, there are two types of dynamics (additive and multiplicative) for $A_1^{(1)}$-surface.
In this case, surface $X$ is multiplicative type, that is, the surface is referred to as the type $q$-$A_1^{(1)}$.
Namely, the birational actions of $\widetilde{W}(E_7^{(1)})$ on parameters $v_i$ and $u_j$ are given by
\begin{align*}
 &s_0:(u_1,u_2,u_3,u_4,u_5,u_6,u_7,u_8,v_1,v_2)\\
 &\quad\mapsto(v_1u_1,v_1u_2,v_1u_3,v_1u_4,v_2u_5,v_2u_6,v_2u_7,v_2u_8,v_2,v_1),\\
 &s_i:(u_i,u_{i+1})\mapsto (u_{i+1},u_i),\quad i=1,2,3,5,6,7,\\
 &s_4:(u_4,u_5,v_1,v_2)\mapsto(u_5,u_4,{u_4}^{1/2}{u_5}^{-1/2}v_1,{u_4}^{-1/2}{u_5}^{1/2}v_2),\\
 &\sigma:(u_1,u_2,u_3,u_4,u_5,u_6,u_7,u_8)
 \mapsto({u_8}^{-1},{u_7}^{-1},{u_6}^{-1},{u_5}^{-1},{u_4}^{-1},{u_3}^{-1},{u_2}^{-1},{u_1}^{-1}),
\end{align*}
and those on the variables $f$ and $g$ are given by
\begin{align*}
 &s_0:(f,g)\mapsto (g,f),\\
 &s_4(f)
 ={u_4}^{1/2}{u_5}^{1/2}
 \cfrac{(v_1-v_2)fg+(v_1u_4-v_2u_5)f+(-u_4+u_5)g}{(-u_4+u_5)f+(v_1u_4-v_2u_5)g+(v_1-v_2)u_4u_5},\\
 &\sigma:(f,g)\mapsto(f^{-1},g^{-1}).
\end{align*}
Note that 
\begin{equation}
 q=\cfrac{{v_1}^2{u_1}^{1/2}{u_2}^{1/2}{u_3}^{1/2}{u_4}^{1/2}}{{u_5}^{1/2}{u_6}^{1/2}{u_7}^{1/2}{u_8}^{1/2}}
\end{equation} 
is invariant under the action of $\widetilde{W}(E_7^{(1)})$.

Let us consider the translation $T_0={R_0}^2$ where
\begin{equation}
 R_0=s_4 s_0 s_3 s_4 s_2 s_3 s_2 s_3 s_4 s_5 
 s_6 s_0 s_4 s_3 s_5 s_4 s_0 s_7 s_6 s_5 s_4 s_1
 s_2 s_4 s_3 s_0 s_4.
\end{equation}
The translation $T_0$ acts 
on the simple roots as
\begin{equation}
 T_0:\bm{\alpha} \mapsto \bm{\alpha}+(0,0,0,-\delta,0,0,0,2\delta),
\end{equation}
where $\bm{\alpha}=(\alpha_1,\dots,\alpha_7,\alpha_0)$,
on the parameters as
\begin{equation}
 T_0:\left(\begin{matrix}u_1&u_2&u_3&u_4\\u_5&u_6&u_7&u_8\end{matrix}\,;~v_1\right)
 \mapsto \left(\begin{matrix}q^{-1}u_1&q^{-1}u_2&q^{-1}u_3&q^{-1}u_4\\qu_5&qu_6&qu_7&qu_8\end{matrix}\,;~q^2v_1\right),
\end{equation}
and on the variables as
\begin{subequations}\label{eqns:qPA1_T0}
\begin{align}
 &\cfrac{(fT_0(g)-q^2v^2)(fg-v^2)}{(fT_0(g)-1)(fg-1)}
 =\cfrac{(f-b_1 v)(f-b_2 v)(f-b_3 v)(f-b_4 v)}{(f-b_5)(f-b_6)(f-b_7)(f-b_8)},\\
 &\cfrac{(fg-v^2)({T_0}^{-1}(f)g-q^{-2}v^2)}{(fg-1)({T_0}^{-1}(f)g-1)}
 =\cfrac{(g-{b_1}^{-1} v)(g-{b_2}^{-1} v)(g-{b_3}^{-1} v)(g-{b_4}^{-1} v)}{(g-{b_5}^{-1})(g-{b_6}^{-1})(g-{b_7}^{-1})(g-{b_8}^{-1})},
\end{align}
\end{subequations}
where
\begin{equation}
 v=v_1,\quad b_i={v_1}^{1/2}u_i\quad (i=1,\dots,4),\quad b_j=\cfrac{u_j}{{v_1}^{1/2}}\quad (j=5,\dots,8).
\end{equation}
The system \eqref{eqns:qPA1_T0} is usually referred to as a $q$-Painlev\'e equation of $A_1^{(1)}$-surface type\cite{MSY2003:MR1958273,SakaiH2001:MR1882403}.

In \cite{HHNS:2015JPhysA}, the discrete dynamical system corresponding to 
\begin{equation}
 \varphi
 =\sigma s_3s_5s_4s_5s_0s_3s_4s_0s_6s_5s_2s_3s_1s_2s_7s_4
 s_3s_2s_5s_6s_4s_5s_3s_4s_3s_7s_6s_7s_1s_2s_1s_0
\end{equation}
is considered with the following special case of parameters:
\begin{subequations}
\begin{align}
 &{u_1}^{1/2}={\rm e}^{-\gamma_o/2},\quad
 {u_2}^{1/2}=-{\rm i}{\rm e}^{-\gamma_o/2},\quad
 {u_3}^{1/2}={\rm e}^{z/2},\quad
 {u_4}^{1/2}=-{\rm i}{\rm e}^{z/2},\\
 &{u_5}^{1/2}={\rm e}^{-z/2},\quad
 {u_6}^{1/2}=-{\rm i}{\rm e}^{-z/2},\quad
 {u_7}^{1/2}={\rm e}^{\gamma_o/2},\quad
 {u_8}^{1/2}=-{\rm i}{\rm e}^{\gamma_o/2},\\
 &v_1={\rm e}^{-z+\gamma_e},\quad
 v_2={\rm e}^{z-\gamma_e},\quad
 q={\rm e}^{2(\gamma_e-\gamma_o)}.
\end{align}
\end{subequations}
The resulting equation is given by
\begin{subequations}\label{eqns:HHNS}
\begin{align}
 &\bar{f}
 =\cfrac{{\rm sinh}(2z)\,fg-{\rm sinh}(z-\gamma_e)-{\rm sinh}(z+\gamma_e)\,g^2}
 {{\rm sinh}(z-\gamma_e)\,fg^2-{\rm sinh}(2z)\,g+{\rm sinh}(z+\gamma_e)\,f}\,,\\
 &\bar{g}
 =\cfrac{{\rm sinh}(2(z+\gamma_e-\gamma_o))\,g\bar{f}-{\rm sinh}(z+\gamma_e-2\gamma_o)-{\rm sinh}(z+\gamma_e)\,\bar{f}^2}
 {{\rm sinh}(z+\gamma_e-2\gamma_0)\,g\bar{f}^2-{\rm sinh}(2(z+\gamma_e-\gamma_o))\,\bar{f}+{\rm sinh}(z+\gamma_e)\,g}\,,
\end{align}
\end{subequations}
where
\begin{equation}
 \varphi:(z,\gamma_e,\gamma_o,f,g)\mapsto
 (z+2(\gamma_e-\gamma_o),\gamma_e,\gamma_o,\bar{f},\bar{g}).
\end{equation}
System \eqref{eqns:HHNS} is another non-autonomous form of Equation \eqref{eqn:aut_Q4eqn} when $k=1$.
In fact, letting 
\begin{equation}
 \gamma_e=\gamma_o=\gamma,
\end{equation}
and substituting 
\begin{align}
 &f=u_{2n-1},\quad
 g=u_{2n},\quad
 \bar{f}=u_{2n+1},\quad
 \bar{g}=u_{2n+2},\\
 &\cfrac{{\rm sinh}(z-\gamma)}{{\rm sinh}(2z)}
 =-\cfrac{\sn(\alpha-\beta)\sn(\alpha)\sn(\beta)}{\sn(\alpha)-\sn(\beta)}\,,\quad
 \cfrac{{\rm sinh}(z+\gamma)}{{\rm sinh}(2z)}
 =\cfrac{\sn(\alpha-\beta)}{\sn(\alpha)-\sn(\beta)}\,,
\end{align}
in System \eqref{eqns:HHNS}, we obtain Equation \eqref{eqn:aut_Q4eqn} with $k=1$.

\section{Proof of Lemma \ref{lemma:symmetry}}\label{section:proof_symmetry}
In this section, we prove Lemma \ref{lemma:symmetry}.
First, we recall the action of $W(E_8^{(1)})=\langle w_0,w_1,\dots,w_8\rangle$ 
on the parameters $u_i$, $i=1,\dots,8$, and $b$:
\begin{subequations}
\begin{align}
 &w_1:\,b\mapsto -b,\\
 &w_2(u_1)=u_1-\cfrac{3(2b+u_1+u_2)}{4},\quad
 w_2(u_2)=u_2-\cfrac{3(2b+u_1+u_2)}{4},\\
 &w_2(u_3)=u_3+\cfrac{2b+u_1+u_2}{4},\quad
 w_2(u_4)=u_4+\cfrac{2b+u_1+u_2}{4},\\
 &w_2(u_5)=u_5+\cfrac{2b+u_1+u_2}{4},\quad
 w_2(u_6)=u_6+\cfrac{2b+u_1+u_2}{4},\\
 &w_2(u_7)=u_7+\cfrac{2b+u_1+u_2}{4},\quad
 w_2(u_8)=u_8+\cfrac{2b+u_1+u_2}{4},\\
 &w_2(b)=b-\cfrac{2b+u_1+u_2}{4},\\
 &w_i:\,(u_{i-1},u_i)\mapsto(u_i,u_{i-1}), ~i=3,\dots,7\label{eqn:action_w_i},\\
 &w_0:\,(u_7,u_8)\mapsto(u_8,u_7),\\
 &w_8:\,(u_1,u_2)\mapsto(u_2,u_1).
\end{align}
\end{subequations}
We define the translations of $W(E_8^{(1)})$ after \cite{KMNOY2003:MR1984002} by
\begin{subequations}
\begin{align}
 &T_1=s_{134}s_{156}s_{789}s_{156}s_{134}w_1,
 &&T_2=w_2T_1w_2{T_1}^{-1},\\
 &T_3=w_3T_2w_3{T_2}^{-1},
 &&T_4=w_4T_3w_4{T_3}^{-1},\\
 &T_5=w_5T_4w_5{T_4}^{-1},
 &&T_6=w_6T_5w_6{T_5}^{-1},\\
 &T_7=w_7T_6w_7{T_6}^{-1},
 &&T_0=w_0T_7w_0{T_7}^{-1},\\
 &T_8=w_8T_3w_8{T_3}^{-1},
\end{align}
\end{subequations}
where
\begin{subequations}
\begin{align}
 s_{134}=&w_2w_3w_8w_3w_2,\\
 s_{156}=&w_4w_3w_2w_5 w_4w_3w_8w_3w_4 w_5w_2w_3w_4,\\
 s_{789}
 =&w_6w_5w_4w_3w_2w_1w_7w_6w_5w_4
   w_3w_2w_0w_7w_6w_5w_4w_3w_8w_3\notag\\
   &w_4w_5w_6w_7w_0w_2w_3w_4w_5w_6
   w_7w_1w_2w_3w_4w_5w_6.
\end{align}
\end{subequations}
Note that $T_i$, $i=0,\dots,8$, commute with each other and satisfy
\begin{equation}\label{eqn:relationT}
 {T_1}^2{T_2}^4{T_3}^6{T_4}^5{T_5}^4{T_6}^3{T_7}^2T_0{T_8}^3=1.
\end{equation}
The actions of translations on the parameters are given by
\begin{align*}
 T_1:&b\mapsto b+\lambda,\\
 T_2:&\left(\begin{array}{llll}u_1&u_2&u_3&u_4\\u_5&u_6&u_7&u_8\end{array}b\right)
 \mapsto
 \left(\begin{array}{llll}u_1-\frac{3}{2}\lambda&u_2-\frac{3}{2}\lambda&u_3+\frac{1}{2}\lambda&u_4+\frac{1}{2}\lambda\\
 u_5+\frac{1}{2}\lambda&u_6+\frac{1}{2}\lambda&u_7+\frac{1}{2}\lambda&u_8+\frac{1}{2}\lambda\end{array}
 b-\frac{\lambda}{2}\right),\notag\\
 T_3:&(u_2,u_3)\mapsto(u_2+2\lambda,u_3-2\lambda),\quad
 T_4:(u_3,u_4)\mapsto(u_3+2\lambda,u_4-2\lambda),\\
 T_5:&(u_4,u_5)\mapsto(u_4+2\lambda,u_5-2\lambda),\quad
 T_6:(u_5,u_6)\mapsto(u_5+2\lambda,u_6-2\lambda),\\
 T_7:&(u_6,u_7)\mapsto(u_6+2\lambda,u_7-2\lambda),\quad
 T_0:(u_7,u_8)\mapsto(u_7+2\lambda,u_8-2\lambda),\\
 T_8:&(u_1,u_2)\mapsto(u_1+2\lambda,u_2-2\lambda),
\end{align*}
where $\lambda=\frac{1}{2}\sum_{i=1}^8 u_i$.
We note that 
\begin{equation}
 W(E_8^{(1)})=W(E_8)\ltimes \langle T_0,\dots,T_8\rangle,
\end{equation}
where $W(E_8)=\langle w_1,\dots,w_8\rangle$.

We are now in a position to prove Lemma \ref{lemma:symmetry}.
The goal here is to find out the generators of group $\Omega$.
Since $W(E_8^{(1)})$ is an infinite set, it is not possible to check 
whether each element satisfies \eqref{eqn:symmetry_condition1} under the condition \eqref{eqn:subspace}.
Therefore, we first reduce the number of the elements which should be checked.

Let 
\begin{align}
 &\overline{\Omega}=\langle w_0,w_1,w_2,r_1,w_4,r_2,w_6,r_3,w_8,T_1,T_2,T_8{T_3}^2T_4,T_4{T_5}^2T_6\rangle,\\
 &H=\langle w_1,w_2,w_4,w_6,w_8,r_1,r_2\rangle,
\end{align}
where $r_i$ are defined by \eqref{eqn:def_r123}.
It is obvious that
\begin{equation}
 H\subset\overline{\Omega}\subset\Omega,\quad 
 H\subset W(E_8).
\end{equation}
We define the equivalence relations $\eqOmega$ and $\eqH$ as follows.
\begin{description}
\item[(i)]
Let $w,w'\in W(E_8^{(1)})$.
We say that $w$ is equivalent to $w'$ by $\overline{\Omega}$ and
we write $w\eqOmega w'$ if $w'=\rho_1 w\rho_2$ with elements $\rho_1,\rho_2\in\overline{\Omega}$.
\item[(ii)]
Let $w,w'\in W(E_8)$.
We say that $w$ is equivalent to $w'$ by $H$ and
we write $w\eqH w'$ if $w'=\rho_1 w\rho_2$ with elements $\rho_1,\rho_2\in H$.
\end{description}
Then, we obtain the following lemma.
\begin{lemma}
The generators of $\Omega$ lie in $\WH\ltimes \langle T_3,T_4,T_5,T_7\rangle$.
\end{lemma}
\begin{proof}
Let $x=wT\in W(E_8^{(1)})$ where $w\in W(E_8)$ and $T\in\langle T_0,\dots,T_8\rangle$.
Assume $w\eqH w'$ where $w'\in W(E_8)$.
Then, there exist $\rho_1,\rho_2\in H$ such that 
\begin{equation}
 w=\rho_1w'\rho_2,
\end{equation}
and $T'={T_0}^{k_0}{T_1}^{k_1}{T_2}^{k_2}{T_3}^{k_3}{T_4}^{k_4}{T_5}^{k_5}{T_6}^{k_6}{T_7}^{k_7}{T_8}^{k_8}$
such that
\begin{equation}
 \rho_2T=T'\rho_2
\end{equation}
since the translation subgroup $\langle T_0,\dots,T_8\rangle$ is a normal subgroup of $W(E_8^{(1)})$.
Then the following relation holds:
\begin{align}
 x&=\rho_1 w'
  {T_1}^{k_1-2k_0}
  {T_2}^{k_2-4k_0}
  {T_3}^{k_3-6k_0}
  {T_4}^{k_4-5k_0}
  {T_5}^{k_5-4k_0}
  {T_6}^{k_6-3k_0}
  {T_7}^{k_7-2k_0}
  {T_8}^{k_8-3k_0}
  \rho_2\notag\\
 &\eqOmega w'
  {T_3}^{k_3-6k_0}
  {T_4}^{k_4-5k_0}
  {T_5}^{k_5-4k_0}
  {T_6}^{k_6-3k_0}
  {T_7}^{k_7-2k_0}
  {T_8}^{k_8-3k_0}\notag\\
 &\eqOmega w'
  {T_3}^{k_3-6k_0-2(k_8-3k_0)}
  {T_4}^{k_4-5k_0-(k_8-3k_0)-(k_6-3k_0)}
  {T_5}^{k_5-4k_0-2(k_6-3k_0)}
  {T_7}^{k_7-2k_0},
\end{align}
because of \eqref{eqn:relationT}
and $\rho_1,\rho_2,T_1,T_2,T_8{T_3}^2T_4,T_4{T_5}^2T_6\in\overline{\Omega}$.
Moreover, it is obvious that
\begin{equation}
 x\in\Omega \Leftrightarrow x'\in\Omega,\quad (x,x'\in W(E_8^{(1)}))
\end{equation}
if $x\eqOmega x'$.
Therefore we have completed the proof.
\end{proof}

We next prove a method of finding out the elements of $\Omega$ from $W(E_8)\ltimes \langle T_3,T_4,T_5,T_7\rangle$ below.
\begin{lemma}\label{lemma:PR_test}
Let $wT\in\Omega$ where $w\in W(E_8)$ and $T\in\langle T_3,T_4,T_5,T_7\rangle$.
Then, there exist ${\bm A}\in{\rm GL}_4(\mathbb{R})$ and ${}^t{\bm K}\in\mathbb{Z}^4$ such that
\begin{equation}\label{eqn:PR_test}
 w.{\bm u}= {\bm A}{\bm u}+2\lambda{\bm K},
\end{equation}
where
\begin{equation}
 {\bm u}=\begin{pmatrix}u_1-u_2\\u_3-u_4\\u_5-u_6\\u_7-u_8\end{pmatrix}.
\end{equation}
\end{lemma}
\begin{proof}
Let $x=wT\in \Omega$ where $w\in W(E_8)$ and $T\in\langle T_3,T_4,T_5,T_7\rangle$.
By definition, the actions of $x$ and $T$ on ${\bm u}$ can be given by
\begin{align}
 &x.{\bm u}={\bm A}{\bm u},\\
 &T^{-1}.{\bm u}={\bm u}+2\lambda {\bm K},
\end{align}
where ${\bm A}\in{\rm GL}_4(\mathbb{R})$ and ${}^t{\bm K}\in\mathbb{Z}^4$.
Therefore, the statement follows from
\begin{equation}
 w.{\bm u}
 =x T^{-1}. {\bm u}
 = {\bm A}{\bm u}+2\lambda{\bm K}.
\end{equation}
\end{proof}

We obtain $\WH=\{ 1,\zeta_1,\dots,\zeta_{93}\}$ by using \textsc{MAGMA}\cite{BCP1997:MR1484478}
with the following commands:
\begin{lstlisting}[basicstyle=\ttfamily\footnotesize]
G < w1, w8, w2, w3, w4,w5, w6, w7 > := CoxeterGroup (GrpFPCox, "E8");
H := sub < G | w1, w2, w4, w6, w8, w3*w4*w8*w3, w5*w4*w6*w5 >;
DoubleCosets (G, H, H);
\end{lstlisting}
We checked the actions of $\{{\zeta_i}^{-1}\}_{i=1,\dots,93}$ exhaustively by using the method given in Lemma \ref{lemma:PR_test}.
As a result, we find that only two of them satisfy \eqref{eqn:PR_test}:
\begin{align}
 &{\zeta_5}^{-1}.{\bm u}
 =\begin{pmatrix}1&0&0&0\\0&1&0&0\\0&0&1&0\\0&0&0&-1\end{pmatrix}{\bm u}
 +2\lambda\begin{pmatrix}0\\0\\0\\-2\end{pmatrix},\\
 &{\zeta_{93}}^{-1}.{\bm u}
 =\begin{pmatrix}0&1&0&0\\0&0&0&-1\\1&0&0&0\\0&0&-1&0\end{pmatrix}{\bm u}
 +2\lambda\begin{pmatrix}0\\-1\\0\\-1\end{pmatrix},
\end{align}
where
\begin{align}
 \zeta_5=&
 w_7 w_6 w_5 w_4 w_3 w_8 w_2 w_1 w_3 w_2
 w_4 w_3 w_8 w_5 w_4 w_3 w_2 w_1 w_6 w_5
 w_4 w_3 w_8 w_2 w_3 w_4\notag\\
 &w_5 w_6 w_7 w_6 w_5 w_4 w_3 w_8 w_2 w_1 
 w_3 w_2 w_4 w_3 w_8 w_5 w_4 w_3 w_2 w_1
 w_6 w_5 w_4 w_3 w_8 w_2\notag\\
 &w_3 w_4 w_5 w_6 w_7,\\
 \zeta_{93}=&
 w_3 w_8 w_2 w_1 w_3 w_2 w_5 w_4 w_3 w_8
 w_6 w_5 w_4 w_3 w_2 w_1 w_7 w_6 w_5 w_4
 w_3 w_8 w_2 w_1 w_3 w_2\notag\\
 &w_4 w_3 w_8 w_5 w_4 w_3 w_2 w_1 w_6 w_5
 w_4 w_3 w_8 w_2 w_3 w_4 w_5 w_7.
\end{align}
Therefore, we obtain the following two elements of $\Omega$:
\begin{equation}
 {\zeta_5}^{-1}T_4{T_5}^2{T_7}^{-2},\quad
 {\zeta_{93}}^{-1}T_4T_5{T_7}^{-1}.
\end{equation}
However, these can be expressed by the elements of $\overline{\Omega}$ as follows
\begin{align}
  &{\zeta_5}^{-1}T_4{T_5}^2{T_7}^{-2}
 =w_0r_2{\zeta_{93}}^{-1}T_4T_5{T_7}^{-1}r_2{\zeta_{93}}^{-1}T_4T_5{T_7}^{-1},\\
 &{\zeta_{93}}^{-1}T_4T_5{T_7}^{-1}
 =w_0w_4w_6w_8 r_1r_2{\varphi_s}^{-1}r_2{\varphi_s}^{-1}
 r_2{\varphi_s}^{-1}r_2{\varphi_s}^{-1}r_1r_2
 r_1\varphi_s r_2 \varphi_s r_2r_1,
\end{align}
where $\varphi_s$ is given by \eqref{eqn:wr_phi}.
Therefore, we have $\Omega=\overline{\Omega}$.
Finally, 
\begin{equation}
 \overline{\Omega}=\langle w_0,w_1,w_2,r_1,w_4,r_2,w_6,r_3,w_8\rangle,
\end{equation}
follows from
\begin{align}
 &T_1=w_2{\varphi_a}^{-1}w_2w_1w_2\varphi_aw_2w_1,\\
 &T_2={\varphi_a}^{-1}w_2\varphi_aw_2,\\
 &T_8{T_3}^2T_4=w_0w_4w_6w_8 w_2{\varphi_a}^{-1}w_2w_1w_2{\varphi_a}^{-1}w_2w_1\varphi_a,\\
 &T_4{T_5}^2T_6=r_2{\varphi_a}^{-1}r_2\varphi_a,
\end{align}
where $\varphi_a$ is given by \eqref{eqn:wr_phi}.
Therefore we have completed the proof of Lemma \ref{lemma:symmetry}.

\def\cprime{$'$}

\end{document}